\documentclass[11pt,3p]{elsarticle}

\usepackage{amsfonts}
\usepackage{amsthm}
\usepackage{amsmath}
\usepackage{times}
\usepackage{lineno}
\usepackage{multirow,array}
\usepackage{hyperref}
\usepackage{tikz}
\usepackage{cleveref}
\usepackage{todonotes}
\usepackage{verbatim}

\newtheorem{theorem}{Theorem}
\newtheorem{proposition}[theorem]{Proposition}
\newtheorem{lemma}[theorem]{Lemma}
\newtheorem{corollary}[theorem]{Corollary}
\newtheorem{definition}{Definition}
\newtheorem{example}{Example}

\newtheorem{conjecture}{Conjecture}
\newtheorem{remark}{Remark}

\renewcommand{\epsilon}{\varepsilon}

\renewcommand{\epsilon}{\varepsilon}

\begin{document}

\sloppy


\begin{frontmatter}

\title{Primitive Sets of Words}

\author[label1]{Giuseppa Castiglione\corref{cor1}}
\ead{giuseppa.castiglione@unipa.it}

\author[label1]{Gabriele Fici}
\ead{gabriele.fici@unipa.it}

\author[label1]{Antonio Restivo}
\ead{antonio.restivo@unipa.it}

\address[label1]{Dipartimento di Matematica e Informatica, Universit\`a di Palermo, Palermo, Italy}

\cortext[cor1]{Corresponding author.}

\journal{Theoretical Computer Science}

\begin{abstract}
Given a (finite or infinite) subset $X$ of the free monoid $A^*$ over a finite alphabet $A$, the rank of $X$ is the minimal cardinality of a set $F$ such that $X \subseteq F^*$.
We say that a submonoid $M$ generated by $k$ elements of $A^*$ is {\em $k$-maximal} if there does not exist another submonoid generated by at most $k$ words containing $M$. We call a  set $X \subseteq A^*$ {\em primitive} if it is the basis of a $|X|$-maximal submonoid. This definition encompasses the notion of primitive word --- in fact, $\{w\}$ is a primitive set if and only if $w$ is a primitive word. 
By definition, for any set $X$, there exists a primitive set $Y$ such that $X \subseteq Y^*$. We therefore call $Y$ a {\em primitive root} of $X$. As a main result, we prove that if a set has rank $2$, then it has a unique primitive root. 
To obtain this result, we prove that the intersection of two $2$-maximal submonoids is either the empty word or a submonoid generated by one single primitive word.

For a single word $w$,  we say that the set $\{x,y\}$ is a {\em bi-root} of $w$ if $w$ can be written as a concatenation of copies of $x$ and $y$ and $\{x,y\}$ is a primitive set. We prove that every primitive word $w$ has at most one bi-root $\{x,y\}$ such that $|x|+|y|<\sqrt{|w|}$. That is, the bi-root of a word is unique provided the  word is sufficiently long with respect to the size (sum of lengths) of the root.

Our results are also compared to previous approaches that investigate pseudo-repetitions, where a morphic involutive function $\theta$ is defined on $A^*$. In this setting, the notions of $\theta$-power, $\theta$-primitive and $\theta$-root are defined, and it is shown that any word has a unique $\theta$-primitive root. This result can be obtained with our approach by showing that a word $w$ is $\theta$-primitive if and only if $\{w, \theta(w)\}$ is a primitive set.
\end{abstract}


\end{frontmatter}

\section{Introduction}

The notion of {\em rank} plays an important role in combinatorics on words. Given a subset $X$ of the free monoid $A^*$ over a finite alphabet $A$, the rank of $X$, in symbols $r(X)$, is defined as the smallest number of words needed to express all words of $X$, i.e., as the minimal cardinality of a set $F$ such that $X \subseteq F^*$. Notice that this minimal set $F$ may not be unique. For instance, the set $X = \{aabca,aa,bcaaa\}$ has rank $2$ and there exist two distinct sets $F_1 = \{aa,bca\}$ and $F_2 = \{a,bc\}$ such that $X \subseteq F_1^*$ and $X \subseteq F_2^*$. It is worth noticing that $r(X) \leq \min\{|X|,|A|\}$, hence $r(X)$ is always finite even if $X$ is an infinite set.
A set $X$ is said to be {\em elementary} if $r(X) = |X|$. The notion of rank -- and the related notion of elementary set -- have been investigated in several papers (cf.~\cite{Neraud90,Neraud92,Neraud93}). In particular, in~\cite{Neraud90} it is shown that the problem to decide whether a finite set is elementary is co-NP-complete.

In this paper, we introduce the notion of primitiveness for a {\em set} of words, which is closely related to that of rank. We first define the notion of {\em $k$-maximal} submomoid. A submonoid $M$ of $A^*$, generated by $k$ elements, is $k$-maximal if there does not exist another submonoid generated by at most $k$ words containing $M$. We then call a set $X \subseteq A^*$ {\em primitive} if it is the basis of a $|X|$-maximal submonoid. Notice that if $X$ is primitive, then $r(X) = |X|$, i.e., $X$ is elementary. The converse is not in general true: there exist elementary sets that are not primitive. For instance, the set $F_{1}=\{aa,bca\}$ is elementary, but it is not primitive since $F_{1}^{*} \subseteq F_{2}^*=\{a,bc\}^*$. The set $F_{2}$, instead, is primitive. 

The notion of primitive set can be seen as an extension of the classical notion of primitive word. Indeed, given a word $w \in A^*$, the set $\{w\}$ is primitive if and only if the word $w$ is primitive. For instance, the set $\{abab,abababab\}$ is not elementary; the set $\{abab\}$ is elementary but not primitive; the set $\{ab\}$ is primitive.

We have from that definition that for every set $X$, there exists a primitive set $Y$ such that $X \subseteq Y^*$. The set $Y$ is therefore called a \emph{primitive root} of $X$. However, the primitive root of a set is not, in general, unique. Consider for instance the set $X = \{abcbab,abcdcbab,abcdcdcbab\}$. It has rank $3$, hence it is elementary, yet it is not primitive. Indeed, $X \subseteq \{ab,cb,cd\}^*$. The set $\{ab,cb,cd\}$ is primitive, and it is a primitive root of $X$. However, it is not the only primitive root of $X$: the set $\{abc,dc,bab\}$ is primitive and $X \subseteq \{abc,dc,bab\}^*$, hence $\{abc,dc,bab\}$ is another primitive root of $X$. In the special case of sets of rank $1$, clearly these always have a unique primitive root. For instance, the primitive root of the set $\{abab,abababab\}$ is the set $\{ab\}$.

As a main result, we prove that if a set has rank $2$, then it has a unique primitive root. This is equivalent to say that for every pair of nonempty words $\{x,y\}$ such that $xy\neq yx$ there exists a unique primitive set $\{u,v\}$ such that $x$ and $y$ can be written as concatenations of copies of $u$ and $v$. The proof is based on the algebraic properties of $k$-maximal submonoids of a free monoid. 

In this investigation, we also take into account another notion of rank, that of {\em free rank} (in the literature, in order to avoid ambiguity, the notion of rank we gave above is often referred to as the {\em combinatorial rank}). The free rank of a set $X$ is the cardinality of the basis of the minimal {\em free} submonoid containing $X$. Closely related to the notion of free rank is the {\em defect theorem}, which states that if $X$ is not a {\em code} (i.e., $X^*$ is not a free submonoid), then the free rank of $X$ is strictly smaller than its cardinality. We are specially interested in the case $k=2$ (that is, the case of $2$-maximal submonoids) and we use the fact that, in this special case, the notions of free rank and (combinatorial) rank coincide. A fundamental step in our argument is Theorem \ref{th_case_2}, which states that the intersection of two $2$-maximal submonoids is either the empty word or a submonoid generated by one primitive word. As a consequence, for every submonoid $M$ generated by two words that do not commute (i.e., two words $\{x,y\}$ such that $xy\neq yx$), there exists a unique $2$-maximal submonoid containing $M$. This is equivalent to the fact that every set of rank $2$ has a unique primitive root. One of the examples we gave above shows that this result is no longer true, in general, for sets of rank $3$ or more --- this highlights the very special role of sets of rank $1$ or $2$.

From these results we derive some consequences on the combinatorics of a {\em single} word. Given a word $w$,  we say that $\{x,y\}$ is a \emph{bi-root} of $w$ if $w$ can be written as a concatenation of copies of $x$ and $y$ and $\{x,y\}$ is a primitive set. We prove that every primitive word $w$ has at most one bi-root $\{x,y\}$ such that $|x|+|y|<\sqrt{|w|}$. That is, the bi-root of a word is unique provided the length of the word is sufficiently large with respect to the size of the root. The notion of bi-root of a single word may be seen as a way to capture a hidden ``repetitive structure'', which encompasses the classical notion of integer repetition (non-primitive word). Indeed, the existence in a word $w$ of a ``short'' (with respect to $|w|$) bi-root reveals some hidden repetition in the word.

As described in the last section, our results can  also be compared to previous approaches that investigate {\em pseudo-repetitions}, where an involutive  morphism (or antimorphism) $\theta$ is defined on the set of words $A^*$. This idea stems from  the seminal paper of Czeizler, Kari and Seki~\cite{CKS}, where originally $\theta$ was the Watson-Crick complementarity function and the motivation was the discovery of hidden repetitive structures in biological sequences.   A word $w$ is called a $\theta$-power if there exists a word $v$ such that $w$ can be factored using copies of $v$ and $\theta(v)$ --- otherwise the word $w$ is called $\theta$-primitive; if $v$ is a $\theta$-primitive word, then it is called the $\theta$-primitive root of $w$. Of course, since the same applies to the word $\theta(w)$, these definitions can be given in terms of the pair $\{w,\theta(w)\}$ and considering as the root the pair $\{v,\theta(v)\}$. With our results, we generalize this setting by considering as a root any pair of words $\{x,y\}$, i.e., dropping the relation between the components of the pair.

Some of the results contained in this paper were presented by the third author in an invited talk at WORDS 2019~\cite{words}. However, this paper significantly differs from the version published in the proceedings of the conference, both in the exposition and in the results presented.

\section{Preliminaries}\label{sec:prem}

Given a finite nonempty set $A$, called the \emph{alphabet}, with $A^*$ (resp.~$A^+=A^*\setminus\{\epsilon\}$) we denote the {\em free monoid} (resp.~{\em free semigroup}) generated by $A$ (under concatenation), i.e., the set of all finite words (resp.~all finite nonempty words) over $A$.

The length $|w|$ of a word $w\in A^*$ is the number of its symbols. The length of the empty word $\epsilon$ is $0$. For a word $w=uvz$, with $u,v,z\in A^*$, we say that $v$ is a \emph{factor} of $w$. Such a factor is called \emph{internal} if $u,z\neq \epsilon$, a \emph{prefix} if $u=\epsilon$, or a  \emph{suffix} if $z=\epsilon$. A word $w$ is \emph{primitive} if $w=v^n$ implies $n=1$, otherwise it is called \emph{a power}. Equivalently, $w$ is primitive if and only if it is not an internal factor of $w^2$.

It is well known in combinatorics on words (see, e.g., \cite{Lothaire1}) that given two words $x$ and $y$ we have $xy=yx$ if and only if $x$ and $y$ are powers of the same word. In this case we say that words  $x$ and $y$  \emph{commute}. As a consequence, a primitive word cannot be written as the concatenation of two words that commute.

Given a subset $X$ of $A^*$, we let $X^*$ denote the submonoid of $A^*$ generated by $X$. Conversely, given a submonoid $M$ of $A^*$, there exists a unique set $X$ that generates $M$ and is minimal for set inclusion.
In fact, $X$ is the set
\begin{equation}\label{base}
    X=(M \setminus \{\epsilon\}) \setminus (M \setminus \{\epsilon\})^2,
\end{equation}
i.e., $X$ is the set of nonempty words of $M$ that cannot be written as a product of two nonempty words of $M$. The set $X$ will be referred to as the {\em minimal generating set} of $M$, or the set of {\em generators} of $M$.

Let $M$ be a submonoid of $A^*$ and $X$ its minimal generating set.
$M$ is said to be {\em free} if any word of $M$ can be {\em uniquely} expressed as a product of elements of $X$. The minimal generating set of a free submonoid $M$ of $A^*$ is called a {\em code}; it is referred to as \emph{the basis} of $M$.
It is easy to see that a set $X$ is a code if and only if, for every $x,y \in X$, $x\neq y$, one has $xX^* \cap yX^* = \emptyset$.
We say that $X$ is a {\em prefix code} (resp.~a {\em suffix code}) if for all $x,y \in X$, one has $x \cap yA^* = \emptyset$ (resp.~$x \cap A^*y = \emptyset$). A code is a {\em bifix code} if it is both a prefix and a suffix code. 
It follows from elementary automata theory that if $X$ is a prefix code, then there exists a  DFA $\mathcal{A}_X$ recognizing $X^*$ whose set of states $Q_X$ verifies (cf.~\cite{BPR}):  $$\vert Q_X \vert \leq \sum_{x \in X} \vert x \vert - \vert X \vert +1.$$

A submonoid $M$ of $A^*$ is called {\em pure} (cf. \cite{Restivo74}) if for all $w\in A^*$ and $n \geq 1$,
$$w^n \in M \Rightarrow w \in M.$$

A set $X \subseteq A^*$ is said to be a {\em circular code} (cf.~\cite{BPR}) if for every $u,v \in A^*$, $uv,vu \in X^*$ implies $u,v \in X^*$.

\noindent Remark that the submonoids of the form $X^*$, with $X$ a circular code, give rise to a special subclass of pure submonoids and are also called {\em very pure} submonoids (cf. \cite{Restivo74}). 

By a result of Tilson~\cite{Tilson}, any nonempty intersection of free submonoids of $A^*$ is free. As a consequence, for any subset $X \subseteq A^*$, there exists the smallest free submonoid containing $X$.

Here we mention the well-known Defect Theorem (cf.~\cite{BPPR79}, ~\cite[Chap.~1]{Lothaire1}, \cite[Chap.~6]{Lothaire2}), a fundamental result in the theory of codes that provides a relation between a given subset $X$ of $A^*$ and the basis of the minimal free submonoid containing $X$ (called the {\em free hull} of $X$).

\begin{theorem}[Defect Theorem]\label{thm:defect}
Let $X$ be a finite nonempty subset of $A^*$. Let $Y$ be the basis of the free hull of $X$.
Then either $X$ is a code, and $Y=X$, or
$$|Y| \leq |X| -1.$$
\end{theorem}

As in \cite{HarjuK04}, given a set $X \subseteq A^*$, we let $r_f(X)$ denote the cardinality of the basis of the free hull of $X$, called the {\em free rank} of $X$. Notice that for any subset $X \subseteq A^*$, $X$ and $X^*$ have the same free rank.  Furthermore, by $r(X)$ we denote the {\em combinatorial rank} (or simply rank) of $X$, defined by:
$$r(X)=\min\{|Y| \mid Y \subseteq A^*, X \subseteq Y^*\}.$$
With this notation, the Defect Theorem can be stated as follows.

\begin{theorem}
Let $X$ be a finite nonempty subset of $A^*$. Then $r_f(X) \leq |X|$, and the equality holds if and only if $X$ is a code.
\end{theorem}
Note that, for any $X \subseteq A^+$, one has
$$r(X) \leq r_f(X) \leq |X|.$$

\begin{example}
Let $X= \{aa,ba,baa\}$. One can prove that $X$ is a code, hence we have $r_f(X)=3$, while $r(X)=2$ since $X\subset \{a,b\}^*$. For $X= \{aa, aaa\}$, we have $r(X)=r_f(X)=1.$
\end{example}

\begin{remark}\label{rank2}
If $\vert X\vert=2$ then $r_f(X)=r(X)$. So for sets of cardinality $2$ we will not specify if we refer to the free rank or to the (combinatorial) rank.

Moreover, from the complexity point of view, N\'eraud proved that deciding if a set has rank $2$ can be done in polynomial time \cite{Neraud93}, whereas for general rank $k$ it is an NP-hard problem~\cite{Neraud90}.
\end{remark}

The {\em dependency graph} (cf.~\cite{HarjuK04}) of a finite set $X\subset A^+$ is the graph $G_X=(X,E_X)$ where $E_X=\{(u,v) \in X \times X \mid uX^* \cap vX^* \neq \emptyset\}$.
Notice that if $X$ is a code, then $G_X$ has no edge. Furthermore, if $(u,v)$ is an edge, then $u$ is a prefix of $v$ or vice versa.
In \cite{Harju1986} and \cite{HarjuK04}, the following useful lemma is proved. 

\begin{lemma}[Graph Lemma] Let $X\subseteq A^+$ be a finite set that is not a code. Then
$$r_f(X)\leq c(X) < |X|,$$
where $c(X)$ is the number of  connected components of $G_X$.
\end{lemma}

\begin{example}
Let $X= \{a,ab, abc, bca, acb, cba\}.$ We have $acba= a \cdot cba= acb \cdot a$ and $abca= a \cdot bca= abc \cdot a$. The basis of the free hull of $X$ is $Y=\{a, ab, bc, cb\}$, hence $r_f(X)=4$. Furthermore, $r(X)=3$ and $c(X)= 4$, as shown in Figure~\ref{fig:1}.
\end{example}

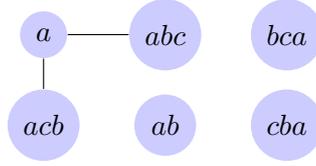
\begin{figure}
\begin{center}
\begin{tikzpicture}
  [scale=.8,auto=left,every node/.style={circle,fill=blue!20}]
  \node (na) at (4,9)  {$a$};
  \node (nabc) at (6,9)  {$abc$};
  \node (nbca) at (8,9)  {$bca$};

  \node (nacb) at (4,7.5) {$acb$};
  \node (nab) at (6,7.5)  {$ab$};
  \node (ncba) at (8,7.5)  {$cba$};

  \foreach \from/\to in {na/nabc}
        \draw (\from) -- (\to);
     \foreach \from/\to in {na/nacb}
    \draw (\from) -- (\to);
\end{tikzpicture}
\end{center}
\caption{\label{fig:1}The dependency graph of $X= \{a, ab, abc, bca, acb, cba\}$.}
\end{figure}

\section{$k$-Maximal Monoids}\label{sec:prim}

With $\mathcal{M}_k$ we denote the family of submonoids of $A^*$ having at most $k$ generators in $A^+$.
The following definition is fundamental for the theory developed in this paper.

\begin{definition}
A submonoid $M \in \mathcal{M}_k$ is $k$-maximal if for every $M^\prime \in \mathcal{M}_k$,  $M \subseteq M^\prime$ implies $M=M^\prime$.
\end{definition}
In other words, $M$ is $k$-maximal if it is not possible to find another submonoid generated by at most $k$ words containing $M$.

\begin{example}
Let $A=\{a,b,c\}$. The submonoid $M=\{a,abca\}^*$ is not $2$-maximal since $abca$ can be factored with $a$ and $bc$, hence $M$ is contained in $\{a,bc\}^*$. On the contrary, $\{a,bc\}^*$ is $2$-maximal since $a$ and $bc$ cannot be factored using two common factors.
\end{example}

\begin{example}\label{3-max}
Let $A=\{a,b,c,d\}$. The submonoid $\{a,cbd,dbd\}^*$ is $3$-maximal, whereas the submonoid $\{a,cbd, dcb\}^*$ is not $3$-maximal since it is contained in $\{a,cb,d\}^*.$ 
\end {example}

\begin{proposition}\label{prop:bifix}
Let $M$ be a $k$-maximal submonoid and $X$ its minimal generating set. Then, $X$ is a bifix code. 
\end{proposition}

\begin{proof}
By contradiction, if $X$ is not prefix (resp. not suffix) then there exist $u,v \in X$ and $t\in A^+$ such that $v=ut$ (resp $v=tu$). It follows that $X^* \subseteq (X  \setminus \{v\} \cup \{t\})^*$, whence $X^*=M$ is not $k$-maximal.
\end{proof}

\begin{remark}
By Proposition \ref{prop:bifix}, it follows that if $X^*$ is $k$-maximal, then $r(X)=r_f(X)=k.$ The inverse implication does not hold in general. For example, the submonoid $\{a,cbd,dcb\}^*$ of Example~\ref{3-max} has both rank and free rank  equal to $3$ and is bifix, but it is not $3$-maximal.
\end{remark}

\begin{proposition}\label{prop:pure}
Let $M$ be a $k$-maximal submonoid. Then $M$ is a pure submonoid.
\end{proposition}
\begin{proof}
 We have to show that, for every $z\in A^*$, if $z^n\in M$, for some $n\geq 1$, then $z\in M$. Let $X$ be the minimal generating set of $M$. If $z^n\in M$, for some $n>1$, then $z\in X$ or the set $X \cup \{z\}$ is not a code. By the Defect Theorem (Theorem~\ref{thm:defect}), there exist $u_1,u_2,\ldots,u_k\in A^+$ such that $(X \cup \{z\})^*\subseteq \{u_1, u_2,...,u_k\}^*$. Since $X^*\subseteq \{u_1, u_2,\ldots,u_k\}^*$ and $X^*$ is $k$-maximal, we have that $X=\{u_1, u_2,\ldots,u_k\}$. Therefore,  $X \cup \{z\} \subseteq X^*$, hence $z\in  X^*$.
\end{proof}

As a direct consequence of Proposition \ref{prop:pure}, we have that a $k$-maximal submonoid is generated by primitive words. However,  not any set of $k$ primitive words generates a $k$-maximal monoid (e.g., $X=\{ab,ba\}^*$ is not $2$-maximal since it is contained in $\{a,b\}^*$).

Submonoids generated by two words, i.e., the elements of $\mathcal{M}_2$, are of special interest for our purposes.
They have been extensively studied in the literature (cf.~\cite{LenSchut,Kar84,Neraud93,Lerest}) and play an important role in some fundamental aspects of combinatorics on words. 

The reader may observe that, as a consequence of some well-known results in combinatorics on words, the submonoids in $\mathcal{M}_1$ have the following important property: If $x^*$ and $u^*$ are two distinct $1$-maximal submonoids (i.e., $x$ and $u$ are primitive words) then $x^* \cap u^* = \{\epsilon\}.$
Next Theorem~\ref{th_case_2}, which represents the main result of this section, can be seen as a generalization of this result to the case of $2$-maximal submonoids.

It is known (see \cite{Kar84}) that if $X$ and $U$ both have rank $2$, then the intersection $X^*\cap U^*$ is a free monoid generated either by at most two words, or by an infinite set of words.  

\begin{example}\label{no2max}
Let $X_1=\{abca,bc\}$ and $U_1=\{a, bcabc\}$. One can verify that $X_1^* \cap U_1^*=\{abcabc, bcabca\}^*.$ Let $X_2=\{aab, aba\}$ and $U_2=\{a, baaba\}$. Then $X_2^* \cap U_2^*=(a(abaaba)^*baaba)^*.$
\end{example}
In the previous example, we have two submonoids that are not $2$-maximal. Indeed, $X_1^*, U_1^* \subseteq \{a, bc\}^*$ and $X_2^*, U_2^* \subseteq \{a, b\}^*.$ We now address the question of finding the generators of the intersection of two $2$-maximal submonoids.

\begin{theorem}\label{th_case_2}
Let $X^*=\{x,y\}^*$ and $U^*=\{u,v\}^*$ be two distinct $2$-maximal submonoids. If $X^*\cap U^*\neq \{\epsilon\}$, then there exists a word $z\in A^+$ such that $X^*\cap U^*=z^*$. Moreover, $z$ is primitive, that is, $X^*\cap U^*$ is $1$-maximal.
\end{theorem}
\begin{proof}
  If $X \cap U =\{z\}$ then $X^*\cap U^*=z^*$. Indeed, if $y=v=z$ and $X^*\cap U^*\neq z^*$ we have the following graph $G_Z$ for $Z=\{x,u,z\}$:
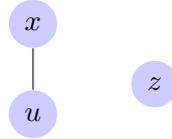
  \begin{center}
\begin{tikzpicture}
  [scale=.8,auto=left,every node/.style={circle,fill=blue!20}]
  \node (nx) at (4,9)  {$x$};
  \node (ny) at (6,8)  {$z$};
  \node (nu) at (4,7.5) {$u$};
  \foreach \from/\to in {nx/nu}
    \draw (\from) -- (\to);
\end{tikzpicture}
\end{center}
 since $\{x,z\}$ and $\{u,z\}$ are bifix sets. Hence, by the Graph Lemma, $r_f(Z) \leq c(Z)=2$, contradicting the $2$-maximality of $X$ and $U$.

  If $X \cap U = \emptyset$, let us consider the set $Z=X \cup U$. We have that $r_f(Z) >2$ since $X^*$ and $U^*$ are $2$-maximal, and, by the Defect Theorem (Theorem~\ref{thm:defect}), $r_f(Z) < 4$ since $Z^*$ is not free (as $X^*\cap U^*$ contains a nonempty word). Hence, the free rank of $Z$ is equal to $3$.

  Let $z$ be an element of the minimal generating set of $X^*\cap U^*$. So, $z=x_1x_2\cdots x_m=u_1u_2\cdots u_n$, with $m,n\geq 1$, $x_i \in X$ and $u_j \in U$, and for every $p<m$ and $q<n$ one has $x_1x_2\cdots x_p\neq u_1u_2\cdots u_q$. Moreover, we can suppose, without loss of generality, that $x_1 =x$ and $u_1 = u$. We want to prove that $z$ is the unique generator of $X^* \cap U^*$. By contradiction, suppose that there exists another  $z'\neq z$ in the minimal generating set of $X^*\cap U^*$, and let $z'=x'_1x'_2\cdots x'_r=u'_1u'_2.\cdots u'_s$.
  If $x'_1 \neq x_1 = x$, then $x'_1 = y$ and we have $xZ^* \cap uZ^* \neq \emptyset$  and $yZ^* \cap u'_1Z^* \neq \emptyset$. In both cases ($u'_1 = u$ or $u'_1 = v$), we have that the graph $G_Z$ has two edges, i.e., $c(Z) = 2$, which is impossible by the Graph Lemma. So, $x_1 = x'_1 = x$. In the same way one can prove that $u_1 = u'_1 = u$ and therefore in the graph $G_Z$ there is only one edge, namely the one joining $x$ and $u$.

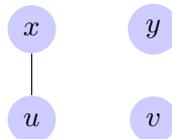
\begin{center}
\begin{tikzpicture}
  [scale=.8,auto=left,every node/.style={circle,fill=blue!20}]
  \node (nx) at (4,9)  {$x$};
  \node (ny) at (6,9)  {$y$};
  \node (nu) at (4,7.5) {$u$};
  \node (nv) at (6,7.5)  {$v$};
  \foreach \from/\to in {nx/nu}
    \draw (\from) -- (\to);
\end{tikzpicture}
\end{center}

Let $h=\max\{i\mid x_j=x'_j\,\,  \forall j \leq i\}$ and $k=\max\{i\mid u_j=u'_j \,\, \forall j \leq i\}$. The hypothesis that $z\neq z'$ implies that $h<m$ and $k<n$. We show that this leads to a contradiction, we then conclude that $z=z'$ is the unique generator of $X^*\cap U^*$.

Without loss of generality, we can suppose that $x_1x_2\cdots x_h$ is a prefix of $u_1u_2\cdots u_k$. Hence, there exists a nonempty word $t$ such that $x_1x_2\cdots x_ht=u_1u_2\cdots u_k$. By definition of $h$, $x_{h+1}\neq x'_{h+1}$, and we can suppose that $x_{h+1}=x$ and $x'_{h+1}=y$. Then,
\begin{equation*}
\begin{split}
t u_{k+1}\cdots u_n & = x_{h+1}\cdots x_m= x\cdots x_m\\
t u'_{k+1}\cdots u'_s& = x'_{h+1}\cdots x'_r= y\cdots x'_r.
\end{split}
\end{equation*}

Set $Z_t=X\cup U\cup \{t\}$. We have
\begin{equation*}
\begin{split}
t Z_t^*\cap xZ_t^* & \neq \emptyset\\
t Z_t^*\cap yZ_t^* & \neq \emptyset.
\end{split}
\end{equation*}
Thus, the graph $G_{Z_t}$ contains the edges depicted in figure:

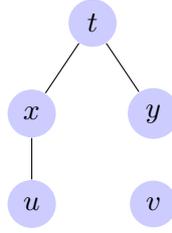
\begin{center}
\begin{tikzpicture}
  [scale=.8,auto=left,every node/.style={circle,fill=blue!20}]
  \node (nt) at (5,10.5) {$t$};
  \node (nx) at (4,9)  {$x$};
  \node (ny) at (6,9)  {$y$};
  \node (nu) at (4,7.5) {$u$};
  \node (nv) at (6,7.5)  {$v$};
  \foreach \from/\to in {nt/nx,nt/ny,nx/nu}
    \draw (\from) -- (\to);
\end{tikzpicture}
\end{center}

By the Graph Lemma, then, the free rank of $Z_t$ is at most $2$, and this contradicts the $2$-maximality of $X^*$ and $U^*$.

Finally, let us prove that $z$ is primitive. Since $X^*$ and $Y^*$ are $2$-maximal, by Proposition \ref{prop:pure} they are both pure, hence also their intersection $z^*$ is pure. But it is immediate that $z^*$ is pure if and only if $z$ is primitive.
 
\end{proof}

\begin{example}\label{exz}
Consider the two $2$-maximal monoids $\{abcab,cb\}^*$ and $\{abc,bcb\}^*$. Their intersection is $\{abcabcbcb\}^*$. The intersection of $\{a,bc\}^*$ and $\{a,cb\}^*$ is $a^*$.
\end{example}

We have shown that the intersection of two $2$-maximal submonoids is generated by at most one element. Moreover, we know that the intersection of two 1-maximal submonoids is the empty word, i.e., it is generated by zero elements. Thus, it is natural to ask if in general, for every $k \geq 1$, the intersection of two $k$-maximal submonoids is generated by at most $k-1$ elements. The following examples, suggested to us by \v{S}t\v{e}p\'an Holub, provide a negative answer to this question.

\begin{example}\label{controes}
The intersection of the two $3$-maximal monoids $\{abc,dc,bab\}^*$ and $\{ab,cb,cd\}^*$ is infinitely generated by  $abc(dc)^*bab$. The intersection of the two $4$-maximal monoids $\{a,b,cd,ce\}^*$ and $\{ac,bc,da,ea\}^*$ is $\{acea,bcea,acda,bcda\}^*$.
\end{example}

Thus, our Theorem \ref{th_case_2} is specific for rank $2$ and cannot be generalized to higher ranks.

For an upper bound on the length of the word that generates the intersection of two $2$-maximal submonoids, we have the following proposition.

\begin{proposition}\label{zbound}
With the hypothesis of Theorem \ref{th_case_2}, $$\vert z \vert < (\vert x \vert + \vert y \vert)(\vert u \vert + \vert v \vert).$$
\end{proposition}
\begin{proof}
Let $\mathcal{A}_X$ (resp. $\mathcal{A}_U$) be the minimal DFA recognizing $X^*$ (resp. $U^*$) and $Q_X$ (resp. $Q_U$) its set of states. Since $X$ and $U$ are bifix codes, we have $\vert Q_X \vert < \vert x \vert + \vert y \vert$ and $\vert Q_U \vert < \vert u \vert + \vert v \vert$. Then the automaton $\mathcal{A}$ recognizing $X^* \cap U^*$ has a set of states $Q$ such that $\vert Q \vert < (\vert x \vert + \vert y \vert)(\vert u \vert + \vert v \vert)$. By Theorem \ref{th_case_2}, $\mathcal{A}$ is composed by only one cycle, labeled by $z$. Thus, $\vert z \vert < (\vert x \vert + \vert y \vert)(\vert u \vert + \vert v \vert).$
 
\end{proof}

Based on our findings, we formulate the following conjecture.
\begin{conjecture}\label{conj}
Let $X^*=\{x,y\}^*$ and $U^*=\{u,v\}^*$ be $2$-maximal submonoids. If $X^*\cap U^*=z^*$, with $z$ primitive, then
$$\vert z \vert < \vert x \vert + \vert y \vert + \vert u \vert + \vert v \vert.$$
\end{conjecture}

\section{Primitive Sets}

We now show how the previous results can be interpreted in the terminology of combinatorics on words.
Let us start with the remark that a word $x\in A^+$ is primitive if and only if, for every $u\in A^+$, $$x^* \subseteq u^* \Rightarrow x=u.$$
With our definition of maximality, we have that a word $x\in A^+$ is primitive if and only if the monoid $x^*$ is $1$-maximal.
Inspired by this observation, we give the following definition.

\begin{definition}
	A finite set $X \subseteq A^*$  is primitive if it is the basis of a $|X|$-maximal submonoid. 
\end{definition}

The following proposition is an easy consequence of the definition of primitive set.
\begin{proposition}
  Any subset of a primitive set is primitive.  
  \begin{proof}
    Let $X \subseteq A^*$ be a primitive set and let $Y$ be a subset of $X$. If $Y$ is not primitive then there exists a set $Z \neq Y$ such that $\vert Z \vert \leq \vert Y \vert$ and $Y \subseteq Z^*.$ It follows that the set $F=X \setminus Y \cup Z$ is such that $\vert F \vert \leq \vert X \vert$ and $X \subseteq F^*$ contradicting the primitiveness of $X$. 
  \end{proof}
\end{proposition}
In particular, any element of a primitive set is a primitive word.

\begin{remark}
The definition of primitive set does not coincide with that of elementary set. A set $X$ is said to be elementary if $r(X)=\vert X \vert $. If $X$ is primitive, then $r(X)=\vert X \vert$, i.e., it is elementary. But there exist elementary sets that are not primitive. For instance, the set $\{aa,bca\}$ is elementary, but it is not primitive since $\{aa,bca\}^{*} \subseteq \{a,bc\}^*$. 
\end{remark}

From the definition of primitive set, we have that for every set $X$ there exists a primitive set $Y$ such that $X \subseteq Y^*$. The set $Y$ is therefore called a \emph{primitive root} of $X$. However, the primitive root of a set is not, in general, unique. Consider for instance the set $X = \{abcbab,abcdcbab,abcdcdcbab\}$. It has rank $3$, hence it is elementary, yet it is not primitive. Indeed, $X \subseteq \{ab,cb,cd\}^*$. The set $\{ab,cb,cd\}$ is primitive, and it is a primitive root of $X$. However, it is not the only primitive root of $X$: the set $\{abc,dc,bab\}$ is primitive and $X \subseteq \{abc,dc,bab\}^*$, hence $\{abc,dc,bab\}$ is another primitive root of $X$. In the special case of sets of rank $1$, clearly these always have a unique primitive root. For instance, the primitive root of the set $\{abab,abababab\}$ is the set $\{ab\}$.

However, as a consequence of Theorem \ref{th_case_2} we have the following result.

\begin{theorem}\label{thm:pair}
A set $X$ of rank $2$ has a unique primitive root.
\end{theorem}
\begin{proof}
	  If $\{u_1, u_2\}$ and $\{v_1, v_2\}$ are two primitive roots of  $X$  then  $X^* \subseteq  \{u_1, u_2\}^* \cap \{v_1, v_2\}^*$. Hence, by Theorem \ref{th_case_2}, $X \subseteq \{z\}^*$, for some primitive word $z$,  i.e., $r(X)= 1$, a contradiction.
\end{proof}

In what follows, we find convenient to call a primitive set of cardinality $2$ a \emph{primitive pair}.

\begin{example}
The words $abca$ and $bc$ are primitive words, yet the pair $\{abca,bc\}$ is not a primitive pair, since $\{abca,bc\}^* \subseteq \{a, bc\}^*$, hence $\{abca,bc\}^*$ is not $2$-maximal.
The pair $\{abcabc, bcabca\}$ can be written as concatenations of copies of both  $\{abca,bc\}$ and $\{a, bcabc\}$.
However, there is a unique way to decompose each word of the pair $\{abcabc, bcabca\}$ as a concatenation of words of a primitive pair, and this pair is $\{a, bc\}$. In fact, the primitive root of $\{abcabc, bcabca\}$ is $\{a, bc\}$.
\end{example}

As it is well known, a primitive word $x$ does not have internal occurrences in $xx$. The next Theorem \ref{thm:cube} provides a similar property in the case of a primitive set of two words.

The following Lemma is a classical result in combinatorics on words, originally due to Lyndon and Sch\"utzenberger \cite{lyndon1962} (cf.~also~\cite{Lothaire1,ChoffKar,LucaV99}).

\begin{lemma}\label{lem:equation}
Let $t,v$ be nonempty words such that $tu=uv$ and $t\neq v$. Then there exists a unique pair of words $(p,q)$ and a unique positive integer $m$ such that $pq$ is primitive and 
\[t=(pq)^m,\ v=(qp)^m,\ u\in (pq)^*p.\]
 \end{lemma}

 \begin{theorem}\label{thm:cube}
 Let $\{x,y\}$ be a primitive pair. Then neither $xy$ nor $yx$ occurs internally in a word of $\{x,y\}^3$.
 \end{theorem}

\begin{proof}
 By symmetry, it is sufficient to prove the statement for $xy$.

 Since $\{x,y\}$ is a primitive pair, we have that both $x$ and $y$ are primitive words. Moreover, also the word $xy$ is primitive. Indeed, $\{x,y\}^*$ is pure, so if $xy=t^n$, $n>1$, then $t\in \{x,y\}^*$.

We will show that for any $w\in \{x,y\}^3$, $xy$ cannot occur internally in $w$. The cases $w=xxx$ and $w=yyy$ are trivial, as $x$ (resp.~$y$) cannot have an internal occurrence in $xx$ (resp.~in $yy$) because $x$ (resp. $y$) is primitive. Let us consider the cases $w=xyx$ and $w=yxy$.  If $xy$ occurs internally in $w=xyx$ (resp.~in $w=yxy$), then so it does in $(xy)^2=wy$ (resp.~in $=xw$), in contradiction with the fact that $xy$ is primitive.

 In the cases $w=xxy$ and $w=xyy$, $x$ (resp.~$y$), would have an internal occurrence in $xx$ (resp.~$yy$), against the primitiveness of $x$ (resp.~of $y$).

 The remaining cases are $w=yxx$, and $w=yyx$. Let us prove the case $w=yxx$.

 Let us first suppose $|y|>|x|$. We have two subcases:
 \begin{enumerate}
 \item \label{it1} the internal occurrence of $y$ does not overlap with the prefix $y$;
 \item\label{it2} the internal occurrence of $y$ overlaps with the prefix $y$.
 \end{enumerate}

\begin{figure}[thb]
\centering

\begin{tikzpicture}[xscale=.6]

\draw [thick](0,0) -- (16,0);
\draw [thick](0,-.1) -- (0, .1);
\draw [thick](6,-.1) -- (6, .1);
\draw [thick](11,-.1) -- (11, .1);
\draw [thick](16,-.1) -- (16, .1);
\node[below] at (3,.5)%
    {$y$};
\node[below] at (8.5,.5)%
    {$x$};
\node[below] at (13.5,.5)%
    {$x$};
\draw [thick]  (4,-1) -- (15,-1);
\draw [thick](4,-.9) -- (4, -1.1);
\draw [thick](9,-.9) -- (9, -1.1);
\draw [thick](15,-.9) -- (15, -1.1);
\node[below left] at (7,-1.1)%
    {$x$};
\node[below left] at (12.5,-1.1)%
    {$y$};

 \node[below] at (5,-.3)%
    {$t$};
\node[below] at (7.5,-.3)%
    {$u$};
    \node[below] at (10,-.3)%
    {$v$};

\draw [dotted] (4,-1) -- (4,0);
\draw[dotted] (6,-1) -- (6,0);
\draw[dotted] (9,-1) -- (9,0);
\draw[dotted] (11,-1) -- (11,0);

\end{tikzpicture}

\caption{Proof of Theorem \ref{thm:cube},  $w=yxx$, $|y|>|x|$, Case \ref{it1}: the internal occurrence of $y$ does not overlap with the prefix $y$.} \label{fig:cube1}
\end{figure}
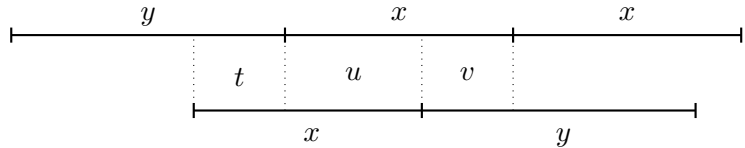

 {\bf Case \ref{it1}}. Since $\{x, y\}$ is a primitive pair, $x$ cannot be a suffix of $y$. Then it follows that $x$ has a non-empty overlap $u$ with itself (see Figure \ref{fig:cube1}).  Therefore, $y$ has a nonempty prefix $v$ and a nonempty suffix $t$ 
 such that $x=tu=uv$. Clearly, $t\neq v$, otherwise $x=tu=ut$ would not be primitive (a word that can be written as the concatenation of two nonempty words that commute is a power of a shorter word). By Lemma~\ref{lem:equation}, we have $t=(pq)^m$, $v=(qp)^m$ and $u\in (pq)^*p$. Now, the internal occurrence of $y$ is a prefix of $vx=vuv$ and it is longer than $x=tu$, so it is of the form $y=vuv'$ for some prefix $v'$ of $v$. Now, $t$ is a suffix of $y$ such that $|t|=|v|>|v'|$. Therefore, since $pq$ cannot occur internally in $pqpq$ (as, by Lemma~\ref{lem:equation}, $pq$ is primitive), and $pq\neq qp$, we have that $v'$ must be of the form $v'=(qp)^iq$ for some $i$. Thus, both $x$ and $y$ belong to $\{p,q\}^*$, against the hypothesis that $\{x, y\}$ is a primitive pair.

 \begin{figure}[htp]
\centering
\begin{tikzpicture}[xscale=.6]

\draw [thick]  (0,0) -- (16,0);
\draw [thick](0,-.1) -- (0, .1);
\draw [thick](8,-.1) -- (8, .1);
\draw [thick](12,-.1) -- (12, .1);
\draw [thick](16,-.1) -- (16, .1);
\node[below] at (4,.5)%
    {$y$};
\node[below] at (10,.5)%
    {$x$};
\node[below] at (14,.5)%
    {$x$};
\draw [thick]  (1,-1) -- (13,-1);
\draw [thick](1,-.9) -- (1, -1.1);
\draw [thick](5,-.9) -- (5, -1.1);
\draw [thick](13,-.9) -- (13, -1.1);
\node[below] at (3,-1.1)%
    {$x$};
\node[below left] at (10,-1.1)%
    {$y$};

 \node[below] at (2.5,-.3)%
    {$t$};
\node[below] at (6.5,-.3)%
    {$u$};
    \node[below] at (10.5,-.3)%
    {$v$};

\draw [dotted] (0,-1) -- (0,0);
\draw[dotted] (5,-1) -- (5,0);
\draw[dotted] (8,-1) -- (8,0);
\draw[dotted] (13,-1) -- (13,0);
\end{tikzpicture}
\caption{Proof of Theorem \ref{thm:cube}, $w=yxx$, $|y|>|x|$, Case \ref{it2}: the internal occurrence of $y$ overlaps with the prefix $y$.}
\label{fig:cube2}
\end{figure}
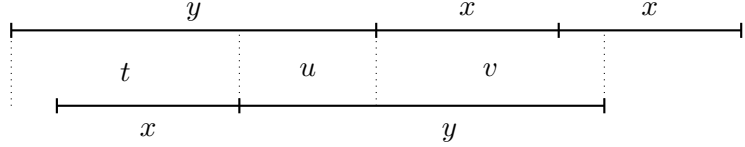

{\bf Case \ref{it2}}.  Let us now suppose that $y$ has an overlap $u$ with itself (see Figure \ref{fig:cube2}). Then we can write $y=tu=uv$, with $u\neq v$ since $y$ is primitive, and in this case $x$ is a suffix of $t$ and a prefix of $v$. By Lemma~\ref{lem:equation}, we have $t=(pq)^m$, $v=(qp)^m$ and $u\in (pq)^*p$. It follows that $x$ has the form $(qp)^iq$. Thus, both $x$ and $y$ belong to $\{p,q\}^*$, against the hypothesis that $\{x, y\}$ is a primitive pair.

\medskip

\begin{figure}[htp]
\centering
\begin{tikzpicture}[xscale=.6]

\draw [thick]  (0,0) -- (19,0);
\draw [thick](0,-.1) -- (0, .1);
\draw [thick](5,-.1) -- (5, .1);
\draw [thick](12,-.1) -- (12, .1);
\draw [thick](19,-.1) -- (19, .1);
\node[below] at (2.5,.5)%
    {$y$};
\node[below] at (8.5,.5)%
    {$x$};
\node[below] at (15.5,.5)%
    {$x$};

\draw [thick]  (3,-1) -- (15,-1);
\draw [thick](3,-.9) -- (3, -1.1);
\draw [thick](10,-.9) -- (10, -1.1);
\draw [thick](15,-.9) -- (15, -1.1);
\node[below] at (6.5,-1.1)%
    {$x$};
\node[below] at (12.5,-1.1)%
    {$y$};

 \node[below] at (4,-.3)%
    {$t$};
\node[below] at (7.5,-.3)%
    {$u$};
    \node[below] at (11,-.3)%
    {$v$};

\draw [dotted] (3,-1) -- (3,0);
\draw[dotted] (5,-1) -- (5,0);
\draw[dotted] (10,-1) -- (10,0);
\draw[dotted] (12,-1) -- (12,0);
\end{tikzpicture}
\caption{Proof of Theorem \ref{thm:cube}, $w=yxx$, $|y|\leq|x|$.}
\label{fig:cube3}
\end{figure}
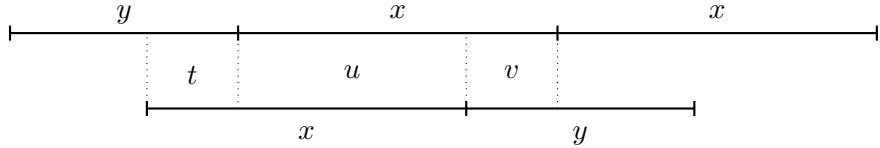

Let now $|y|\leq |x|$.

The internal occurrence of $xy$ must begin before the end of the prefix $y$ of $w$, otherwise, if the occurrence of $xy$ starts exactly where the prefix $y$ ends then $y$ is a prefix of $x$ against the hypothesis of primitive pair, if it starts after then $x$ would have an internal occurrence in $xx$ , against the hypothesis that $x$ is primitive. So, $x$ has an overlap $u$ with itself (see Figure \ref{fig:cube3}).

As in Case \ref{it1}, $y$ has a prefix $v$ and a suffix $t$ 
such that $x=tu=uv$. Clearly, $t\neq v$, otherwise $x=tu=ut$ would not be primitive. By Lemma~\ref{lem:equation}, we have $t=(pq)^m$, $v=(qp)^m$ and $u\in (pq)^*p$. Now, the internal occurrence of $y$ is a prefix of $vx=vuv$ but now it is shorter than $x=vu$, as $|y|\leq |x|$, so it is of the form $y=vu'$ for some prefix $u'$  of $u$. Since $pq$ cannot occur internally in $pqpq$ (as, by Lemma~\ref{lem:equation}, $pq$ is primitive), and $pq\neq qp$, and since $y$ ends in $t$, we have that $u'$ must be of the form $u'=(pq)^i$ for some $i>0$. Thus, both $x$ and $y$ belong to $\{p,q\}^*$, against the hypothesis that $\{x, y\}$ is a primitive pair.

The case $w=yxx$ is now proved.

The proof of the case $w=yyx$ is analogous.
 \end{proof}

\begin{remark}
In the statement of Theorem~\ref{thm:cube}, the hypothesis that $\{x,y\}$ is a primitive pair cannot be replaced by simply requiring that $x$ and $y$ are primitive words. As an example, let $x=abcabca$, $y=bcaabcabc$; the word $xy$ has an internal occurrence in $yxx$, yet $\{x,y\}\subset \{a,bc\}^*$. 

It is also worth noticing that the converse of Theorem~\ref{thm:cube} does not hold, in general. For example, $\{abcaa, bc\}$ is not primitive ($\{abcaa,bc\}^* \subseteq \{a, bc\}^*$), yet neither $abcaabc$ nor $bcabcaa$ occurs internally in a word of $\{x,y\}^3$.
\end{remark}

Moreover, we can infer the following properties.

\begin{corollary}\label{cor:factor}
Let $\{x,y\}$ be a primitive pair and $w$ a primitive word in $\{x,y\}^2\{x,y\}^*$. For all $u,v \in A^*$, if $uwv \in \{x,y\}^*$ then $u,v \in \{x,y\}^*.$
\end{corollary}
\begin{proof}
Let $w=x_1 \cdots x_m$, with $x_i \in \{x, y\}$, be the unique factorization of $w$ in $\{x, y\}$. Since $w$ is primitive, such a factorization necessarily contains $xy$ (or equivalently $yx$), i.e., there exists $1 \leq j < m$ such that $x_j=x$ and $x_{j+1}=y$. If $uwv \in \{x,y\}^*$ with $uwv=y_1\cdots y_n$, $y_i \in \{x, y\}$, then there exist $h,k$, with $1 \leq h < k \leq n$, such that $y_h=x_1, \dots, y_k=x_m$, since otherwise the block $xy$ would appear as internal factor in some word of $\{x, y\}^3$, contradicting Theorem \ref{thm:cube}. Hence, $u=y_1 \cdots y_{h-1}$ and $v=y_{k+1} \cdots x_n$.  
\end{proof}

\begin{proposition}
 If $\{x, y\}$ is a primitive pair, then $\{x, y\}$ is a circular code.
\end{proposition}
\begin{proof} Let $u,v \in A^*$ be words such that $uv, vu \in \{x,y\}^*$. We have to show that $u, v\in \{x,y\}^*$. If $uv=x^n$, we prove that $vu=y^n$. Indeed, $vu$ is a conjugate of a power of $x$, therefore it is a power of a conjugate of $x$. Let $vu=z^n$, with $z$ a conjugate of $x$ and hence primitive. Since $\vert uv \vert = \vert vu \vert$, we have $\vert z \vert = \vert x \vert$. Thus, either $z=x$ or $z=y^n$ (in particular $z=y$ because primitive). If $z=x$, then $uv=vu=x^n$, against the primitiveness of $x$. If $z=y$, we have $vu=y^n$, hence $x$ and $y$ are conjugate. So there exist $p,q$ such that $x=pq$ and $y=qp$ and $\{x,y\}^* \subseteq \{p,q\}^*$, a contradiction with the hypothesis that $\{x,y\}$ is a primitive pair. If $uv$ is not a power of $x$, then its unique factorization in $\{x,y\}$ contains $xy$. We can therefore write $uv=zxyt$, with $z,t \in A^*$. By Corollary \ref{cor:factor}, $z,t \in \{x,y\}^*$. If $u=z$, or $u=zx$, or $u=zxy$, we have done. Otherwise, $uv=zx_1x_2yt$, with $x=x_1x_2$, $u=zx_1$, $v=x_2yt$. In such a case, $vu=x_2ytzx_1$ and $ytz \in \{x,y\}^2\{x,y\}^*$. Thus, by Corollary \ref{cor:factor}, $x_1, x_2 \in \{x,y\}^*$, i.e., either $x_1=x$ and $x_2=\epsilon$ (resp. $x_2=x$ and  $x_1=\epsilon$) or $x_1,x_2 \in y^*$, which implies $x \in y^*$, a contradiction. This concludes the proof.
\end{proof}

\section{Bi-root of a Single Primitive Word}

In this section, we derive  some consequences on the combinatorics of a {\em single} word. In particular, we introduce the notion of bi-root of a primitive word, and we show how this notion may be useful to reveal some hidden repetitive structure in the word.

Let $w$ be a nonempty word. If $w$ is not primitive, then it can be written in a unique way as a concatenation of copies of a primitive word $r$, called the {\em root} of $w$. 
However, if $w$ is primitive, one can ask whether it can be written as a concatenation of copies of two words $x$ and $y$. If we further require that $\{x, y\}$ is a primitive set, then we call $\{x, y\}$ a {\em bi-root} of the word $w$. Note that the bi-root of a single word is not, in general, unique. For instance, for $w=abcbac$ we have $w=ab \cdot cbac= abcb \cdot ac$ and $\{ab, cbac\}$ and $\{abcb, ac\}$ are both primitive pairs, i.e., they are both bi-roots of $w$. However, if we additionally require that the size $\vert x \vert + \vert y \vert$ of the bi-root $\{x,y\}$ is ``short'' with respect to the length of $w$, then we obtain again the uniqueness. This is shown  in the next theorem.

\begin{theorem}\label{thm:squareroot}
	Let $w$ be a primitive word. Then $w$ has at most one bi-root $\{x,y\}$ such that $|x|+|y|<\sqrt[]{|w|}$.
\end{theorem}

\begin{proof}
	Suppose by contradiction there exists another bi-root $\{u,v\}$ of $w$ with $|u|+|v|<\sqrt[]{|w|}$. Take $X=\{x,y\}$ and $U=\{u,v\}$. By Theorem \ref{th_case_2}, there exists a primitive word $z$ and an integer $n$ such that $w=z^n$. As $w$ is primitive, $w=z$ and $n=1$. By Proposition~\ref{zbound},  we have that $|w|<(|x|+|y|)(|u|+|v|)<\sqrt[]{|w|} \cdot \sqrt[]{|w|} = |w|$, a contradiction.
	 
\end{proof}

The following example shows a word $w$ that has bi-roots of different sizes, but only one of size less than $\sqrt{|w|}$.

\begin{example}
	Consider the primitive word $w=abcaabcabc$ of length $10$. The pair $\{a,bc\}$ is the only bi-root of $w$ of size smaller than $\sqrt{|w|}$. 
\end{example}

Asking for a tight bound in the statement of Theorem~\ref{thm:squareroot} is of course a problem intimately related to Conjecture~\ref{conj}.
\begin{conjecture}
Let $w$ be a primitive word. Then $w$ has at most one bi-root $\{x,y\}$ such that $|x|+|y|< \vert w \vert /2$.
\end{conjecture}

We observe that both the classical notion of root and the notion of bi-root are related to some  repetitive structure inside the word. If $w$ is not primitive, the length of its root reveals its repetitive structure in the sense that, if such a length is much smaller than the length of $w$, then the word $w$ can be considered highly repetitive. If $w$ is primitive, the size of its bi-root (intended as the sum of the lengths of the two components of the pair) plays an analogous role. This could be illustrated by the following (negative) example. Consider a word $w$ over the alphabet $A$ such that all the letters of $w$ are distinct, so that $|w|=|A|$. This word is not repetitive at all, and it has $|w|-1$ different bi-roots $\{x,y\}$, all of size $|w|$, corresponding to the trivial factorizations $w=xy$. Thus, the absence of repetitions in a word is related to the large size of its bi-roots. On the contrary, the existence in a word $w$ of a ``short'' (with respect to $|w|$) bi-root corresponds to the existence of some hidden repetitive structure in the word. This approach is connected to some already-considered notions of hidden repetitions (cf.~\cite{GMN13,G13,GaMa19}), as we further discuss in the next section.

We think that the notion of bi-root can be further explored and may have applications, e.g., in the area of  string algorithms.

Notice that the minimal length of a bi-root is affected by the combinatorial properties of the word. For example, if $w$ is a square-free word, then $w$ cannot have a bi-root $\{x,y\}$ such that $|x|+|y|<|w|/4$, since otherwise $w$ would contain a square ($xx$, $yy$, $xyxy$ or $yxyx$).
The previous remark suggests a possible link between the notion of a bi-root and the classical notion of \emph{binary pattern}, which has been deeply investigated in combinatorics on words and fully classified by J.~Cassaigne~\cite{Cas94} (see also~\cite[Chap.~3]{Lothaire2} for a survey). 

\section{Connections with Pseudo-Primitive Words}\label{sec:pseudo}

We now show how the notion of a primitive pair can be seen as a generalization of the notion of a pseudo-primitive word, with respect to an involutive (anti-)morphism $\theta$, as introduced in~\cite{CKS}.

A map $\theta: A^* \rightarrow A^*$ is a {\em morphism} (resp. {\em antimorphism}) if for each $u, v \in A^*$, $\theta(uv)=\theta(u)\theta(v)$ (resp. $\theta(uv)=\theta(v)\theta(u)$); $\theta$ is an {\em involution} if $\theta(\theta(a))=a$ for every $a \in A$. 

Let $\theta$ be an involutive morphism or antimorphism other than the identity function. We say that a word $w \in A^*$ is a {\em $\theta$-power} of $t$ if $w \in t\{t, \theta(t)\}^*$. A word $w$ is {\em $\theta$-primitive} if there exists no nonempty word $t$  such that $w$ is a $\theta$-power of $t$ and $\vert w \vert > \vert t \vert$. 

\begin{theorem}[\cite{CKS}]\label{kari}
	Given a word $w \in A^*$ and an involutive (anti-)morphism  $\theta$, there exists a unique $\theta$-primitive word $u \in A^*$ such hat $w$ is a $\theta$-power of $u$. The word $u$ is called the {\em $\theta$-root} of $w$.
\end{theorem}

\begin{example}\label{ex_thetaroot}
	Let $\theta: \{a,b,c\}^* \rightarrow \{a,b,c\}^*$ the involutive morphism defined by $\theta(a)=b$, $\theta(b)=a$ and $\theta(c)=c$. The $\theta$-root of the word $abcabcbac$ is $abc$.
\end{example}

If $\theta$ is an involutive morphism, we show that Theorem \ref{kari} can be obtained as a consequence of Theorem \ref{thm:pair}. If $\theta$ is an involutive antimorphism, we obtain a slightly different formulation, from which we derive a new property of $\theta$-primitive words.

Given an (anti-)morphism $\theta$ and a set $X \subseteq A^*$, $\theta(X)$ denotes the set $\{\theta(u) \mid u \in X\}$. We say that $X$ is {\em $\theta$-invariant} if $\theta(X) = X$.

We have the following propositions.

\begin{proposition}\label{prop_invariant2}
	Let $\theta$ be an involutive (anti-)morphism. If $\{x, y\}$ is $\theta$-invariant, then so is its root.
\end{proposition}

\begin{proof}

If $\{u,v\}$ is the root of $\{x,y\}$, then  $\theta(\{u,v\})$ is the root of $\theta(\{x,y\})$. Since $\theta(\{x,y\}) = \{x,y\}$ and, by Theorem 9, the root is unique, it follows that $\theta(\{u,v\}) = \{u,v\}$.
\end{proof}

\begin{example} 
	Let $\theta$ be as in Example \ref{ex_thetaroot}. The pair $\{abcabcbac, abcbacabc\}$ is $\theta$-invariant. However, it is not a primitive pair. Its bi-root is the pair $\{abc, bac\}$, which is $\theta$-invariant since $\theta(abc)=bac$.
\end{example}

\begin{remark}
	Let $\theta$ be an involutive morphism. Then $\{x, y\}$ is $\theta$-invariant if and only if $y=\theta(x)$. If $\theta$ is an involutive antimorphism, then $\{x, y\}$ is $\theta$-invariant if and only if either $y=\theta(x)$ or $x=\theta(x)$ and  $y=\theta(y).$ In the last case, $x$ and $y$ are called {\em $\theta$-palindromes}.  
\end{remark}

\begin{example}\label{reverse}
	Let  $\theta: \{a, b, c\}^* \mapsto \{a, b, c\}^*$ be the involutive antimorphism  defined by $\theta(a)=a$, $\theta(b)=b$, $\theta(c)=c$.  The pair $\{abcbbcba, abcba\}$ is $\theta$-invariant. Its bi-root is $\{a, bcb\}$, which is $\theta$-invariant since composed by $\theta$-palindromes. With the same $\theta$, the pair $\{abbbbabba, abbabbbba\}$ is $\theta$-invariant and its bi-root is $\{abb, bba\}$, which is $\theta$-invariant since $\theta(abb)=bba$.
\end{example}

\begin{proposition} \label{prop:tetaprimitive}
	Let $w \in A^*$ and $\theta$ be an involutive morphism of $A^*$. Then, $w$ is $\theta$-primitive if and only if the pair $\{w, \theta(w) \}$ is a primitive pair.
\end{proposition}
\begin{proof}
	Let us suppose, by contradiction, that $\{w, \theta(w) \}$ is a primitive pair and $w$ is not $\theta$-primitive. Then there exists $t$ such that $w \in \{t,\theta(t)\}^*$. Hence, $\theta(w) \in \{t,\theta(t)\}^*$, so the pair $\{w, \theta(w) \}$ is not primitive.  
	Conversely, let us suppose that $w$ is $\theta$-primitive and $\{w, \theta(w)\}$ is not a primitive pair. Denote by $\{u,v\}$ its bi-root. Since $\{w, \theta(w)\}$ is $\theta$-invariant, then $\{u,v\}$ is $\theta$-invariant, i.e., $v=\theta(u)$. Hence, $w \in \{u, \theta(u)\}^*$, i.e., $w$ is not $\theta$-primitive.
	 
\end{proof}

From Theorem \ref{thm:pair} and Proposition \ref{prop:tetaprimitive} we derive Theorem \ref{kari} when $\theta$ is an involutive morphism.

Now, let us consider the case of antimorphisms. Reasoning analogously as we did in the proof of Proposition~\ref{prop:tetaprimitive}, we can prove the following result.

\begin{proposition}
	Let $w \in A^*$ and $\theta$ an involutive antimorphism of $A^*$. If the pair $\{w, \theta(w) \}$ is a primitive pair, then $w$ is $\theta$-primitive.
\end{proposition}

The converse does not hold in general, as the following example shows.

\begin{example}
	Let $\theta$ be the antimorphic involution of Example \ref{reverse}. The word $w=abbaabbacbc$ is $\theta$-primitive, whereas the pair $\{w, \theta(w)\}= \{abbaabbacbc, cbcabbaabba\}$ is not a primitive pair, since its bi-root is the pair $\{abba, cbc\}$.
\end{example}

Finally, we can state the following proposition, which provides a factorization property of $\theta$-primitive words. 

\begin{proposition}
	Let $w \in A^*$ and $\theta$ an involutive antimorphism. If $w$ is $\theta$-primitive and $\{w, \theta(w)\}$ is not a primitive pair, then there exist two $\theta$-palindromes $p$ and $q$ such that $w \in \{p, q\}^*$. 
\end{proposition}
\begin{proof}
	Suppose that $\{w, \theta(w)\}$ is not a primitive pair and denote by $\{u, v\}$ its bi-root. Since $\{w, \theta(w)\}$ is $\theta$-invariant,  then so is $\{u, v\}$ by Proposition \ref{prop_invariant2}, and $v \neq \theta(u)$ since $w$ is $\theta$-primitive. Then, $u=\theta(u)$ and $v=\theta(v)$ are $\theta$-palindromes.
	 
\end{proof}

Finally, we point out that our Theorem \ref{thm:cube} can be viewed as a generalization of the following result of Kari, Masson and Seki \cite{Kari11}:

\begin{theorem}[Theorem 12 of \cite{Kari11}]
	Let $x$ be a nonempty $\theta$-primitive word. Then neither $x\theta(x)$ nor $\theta(x)x$ occurs internally in a word of $\{x,\theta(x)\}^3$. 
\end{theorem}

\section{Conclusions}

We introduced the notion of $k$-maximal submonoid, together with its basis, which we call a primitive set. We showed that the notion of $k$-maximal submonoid allows one to give new results in a classical area of formal languages and theory of codes: the algebraic structure of the intersection of two monoids. In particular, we showed that the intersection of two $2$-maximal monoids is either empty or $1$-maximal, that is, generated by a single primitive word. This result is no longer true, in general, for larger values of $k$ --- the study of the intersection of two $3$-maximal monoids will be the object of a forthcoming paper~\cite{Cas-Hol}.

The notion of primitive set (and the corresponding notion of root) can be viewed as a natural generalization of the fundamental notion of primitive word in combinatorics on words. We exhibited some new structural results that make use of this notion. In particular, we showed that primitive sets can be used to express some kinds of hidden repetitive structures that have been considered in the area of string matching.

\bibliographystyle{elsarticle-num}

\end{document}